\RequirePackage{fix-cm}
\documentclass{svjour3}                     
\AtBeginDocument{ \paperwidth=\dimexpr 1in + \oddsidemargin + \textwidth + 1in + \oddsidemargin \relax \paperheight=\dimexpr 1in + \topmargin + \headheight + \headsep + \textheight + 1in + \topmargin \relax \usepackage[pass]{geometry} \relax }
\usepackage{hyperref}
\usepackage[T1]{fontenc} 
\usepackage{tgtermes, tgheros} 
\usepackage{amssymb} 
\usepackage{amsfonts}

\usepackage{amsthm}
\usepackage{amsmath}
\usepackage{mathtools}
\usepackage{optidef}
\usepackage{stmaryrd}
\usepackage{graphicx, color}
\usepackage{natbib}
\usepackage{caption}
\usepackage{subcaption}
\usepackage{layout}
\usepackage{comment}
\usepackage{setspace}
\usepackage{lscape}
\usepackage{dcolumn}
\newcolumntype{.}{D{.}{. }{4}}
\usepackage{booktabs}
\usepackage{multirow}
\usepackage{threeparttable}
\usepackage{tabularx}
\usepackage[table]{xcolor}
\usepackage{nicematrix}
\usepackage{marvosym}
\usepackage{tocloft}
\newtheorem{propo}{Proposition}
\theoremstyle{definition}
\newtheorem*{dfn}{Definition}
\newtheorem*{asp}{Assumption}
\journalname{ }
\def\makeheadbox{\relax}
\begin{document}
\title{
Nonlinear Domar aggregation over transforming production networks
}
\titlerunning{ }        
\author{Satoshi Nakano \and Kazuhiko Nishimura}
\institute{
Satoshi Nakano \\ 0000-0003-0424-585X 
 \at Nihon Fukushi University, Tokai, 477-0031, Japan \\
\email{nakano@n-fukushi.ac.jp}
\and 
Kazuhiko Nishimura (Corresponding Author) \\ 0000-0002-7400-7227 
 \at Chukyo University, Nagoya, 466-8666, Japan \\
\email{nishimura@lets.chukyo-u.ac.jp}
}
\date{\today}

\maketitle
\begin{abstract}
An economy-wide production network, manifested through monetary input--output coefficients, inherently destabilizes during the general equilibrium propagation of sectoral productivity shocks when substitution elasticities are non-neutral. This study explores the global properties of such networks by mapping the non-linear price manifold into a linearized \textit{transcendent space}. Within this framework, we identify the emergence of network \textit{singularities}, identifying the metabolic thresholds where productivity declines lead to supply-chain paralysis or efficiency gains render primary factors redundant. Furthermore, we demonstrate that the interaction between productivity shocks---the sign of \textit{synergism}---is uniquely determined by the substitution elasticity $\sigma$. Our findings transform industrial policy into an \textit{inverse problem} of network topology: we provide a rigorous justification for why an inelastic network necessitates selective concentration on bottleneck sectors, whereas an elastic network favors a diversified investment strategy.
\end{abstract}

\clearpage
\section{Introduction}

Suppose that all industrial sectors of an economy reveal constant returns to scale. In such a setting, the unit cost of each sector depends uniquely on factor prices and its own productivity level. According to the nonsubstitution theorem \citep{koopmans}, equilibrium prices in such an economy are independent of the scale of production, implying that changes in final consumption do not induce technical substitution. However, this invariance does not extend to productivity fluctuations. As sectors are interconnected through an input--output network, sectoral productivity shocks propagate across the entire system, determining general equilibrium prices  and endogenously reshaping the production technology. The aggregation of these sectoral shocks into a measure of aggregate performance is known as Domar aggregation.

According to \citet{hultenRES1978}, this aggregation is essentially linear: sectoral productivity growth is mapped into aggregate growth using static Domar weights. However, this linearity holds only in a Cobb--Douglas economy ($\sigma=1$), where monetary input--output coefficients remain invariant to shocks. While recent influential studies \citep[e.g.,][]{aceECTA2012, aceAER2017, aceECTA2020} have exploited this linear framework to explain aggregate fluctuations through the ``granularity'' of the network \citep{gabaixECTA2011}, the assumption of neutral substitution limits our understanding of structural transformation. When substitution is non-neutral ($\sigma \neq 1$), the production network becomes an evolving entity; productivity shocks trigger changes in relative prices, which in turn endogenously reshape the monetary input--output linkages through factor substitution.

This transformation renders Domar aggregation inherently nonlinear. While \citet{baqaeeECTA2019} addressed this by using a second-order approximation to capture local nonlinearities, such local methods may overlook the global behavior of the system, particularly when the economy approaches structural \textit{singularities}. In contrast, \citet{nnMD2024} utilized industry-specific estimated elasticities to directly solve the general equilibrium price system through recursive mapping, demonstrating that differences in aggregate volatility are remarkably well-explained by sectoral elasticity profiles. 

The present study advances this line of inquiry by providing a rigorous theoretical foundation for nonlinear Domar aggregation. We move beyond both local approximations and numerical recursive methods to explore the global algebraic properties of the production network using a universal substitution elasticity $\sigma$ and \textit{transcendent variables}. This framework allows us to linearize the price system in transcendent space, enabling an analytical evaluation of the network's \textit{metabolic viability}.

Our contribution is twofold. First, we identify the emergence of network singularities. We show that an inelastic network ($\sigma < 1$) can encounter a ``negative singularity,'' where prices diverge as productivity declines---a state of supply-chain paralysis. Conversely, an elastic network ($\sigma > 1$) can approach a ``positive singularity'' where prices vanish, reflecting a regime where the primary factor is displaced by the self-improving efficiency of intermediate linkages. While \citet{MoranBouchaud2019} identify the stability limits of large economies through random matrices, we uncover the specific aggregation process and internal structural transformation that drive the system toward such limits.

Second, we provide a strategic guide for industrial policy by theorizing \textit{synergism}. We argue that the optimal innovation portfolio---whether to adopt a selective concentration policy or a diversified investment strategy---is an \textit{inverse problem} of the network's topology. By explicitly considering the granularity of the production network \citep{gabaixECTA2011}, we show that the sign of synergism determines the effectiveness of resource allocation. If the network reveals positive synergism ($\sigma > 1$), simultaneous innovations amplify each other, making diversified investment optimal. If it reveals negative synergism ($\sigma < 1$), innovations counteract one another, suggesting that a selective policy focusing on bottleneck sectors is more effective for social welfare.

The rest of the paper is organized as follows. Section 2 develops the general equilibrium system and derives the nonlinear Domar aggregation function. Section 3 analyzes network singularities and the conditions for metabolic viability. Section 4 explores synergism and its implications for industrial policy. Section 5 concludes the paper.

\section{The Model}

In this section, we delineate the formal structure of our general equilibrium system. We consider an economy consisting of $n$ industrial sectors, each characterized by constant returns to scale. The pivotal departure from the conventional literature is our treatment of substitution elasticities as non-neutral, which fundamentally transforms the production network's topology in response to productivity fluctuations.

\subsection{Household and Aggregate Output}
The representative household's utility is governed by a Cobb--Douglas function, with the corresponding indirect utility $V(\boldsymbol{c}(\boldsymbol{p}, B)) = B \prod_{i=1}^n (p_i)^{-\mu_i}$, where $B$ denotes the nominal GDP. By normalizing the reference prices to unity, the aggregate output growth is evaluated as:
\begin{align}
\log V = - \sum_{i=1}^n \mu_i \log p_i = -  (\log \boldsymbol{p})\boldsymbol{\mu}
\label{gdpgrowth}
\end{align}
where $\boldsymbol{\mu}$ denotes the expenditure share vector. This formulation ensures that aggregate performance is intrinsically tied to the general equilibrium price vector $\boldsymbol{p}$.

\subsection{Production and Linearization in Transcendent Space}
Consider the CES production of constant returns to scale and the corresponding unit cost function for the $j$th industry:
\begin{align}
F(\boldsymbol{x}; {z}) = {z} \left( 
\sum_{i=0}^n ({\alpha}_{i})^{\frac{1}{\sigma}}({x}_{i})^{\frac{\sigma-1}{\sigma}}
\right)^{\frac{\sigma}{\sigma-1}},
&&
C(\boldsymbol{p}; {z}) = {z}^{-1} \left( 
\sum_{i=0}^n {\alpha}_{i} ({p}_{i})^{1-\sigma}
 \right)^{\frac{1}{1-\sigma}}
 \label{base_revised}
\end{align}
where $\sigma$ denotes the universal elasticity of substitution and ${z} > 0$ is the productivity level. Under non-neutrality ($\sigma \neq 1$), the zero-profit condition $p = C(\boldsymbol{p}; z)$ represents a non-linear relationship in the price manifold. To resolve this, we map the system into a \textit{transcendent space} by defining the transcendent price $\pi = p^{1-\sigma}$ and the transcendent productivity ${\zeta} = {z}^{\sigma-1}$. The zero-profit condition then becomes linear in terms of these transcendent variables:
\begin{align}
\pi = {\zeta} \left( \alpha_0 + \alpha_1 \pi_1 + \cdots + \alpha_n \pi_n \right)
\label{linear_trans}
\end{align}
The general equilibrium system for the multisector economy is concisely described as:
\begin{align}
\boldsymbol{\pi} = \boldsymbol{\alpha}_0 \langle \boldsymbol{\zeta} \rangle + \boldsymbol{\pi} \langle \boldsymbol{\zeta} \rangle \mathbf{A}
\label{ge_matrix}
\end{align}
The algebraic solution, $\boldsymbol{\pi} = \boldsymbol{\alpha}_0 \langle \boldsymbol{\zeta} \rangle [ \mathbf{I} - \langle \boldsymbol{\zeta} \rangle \mathbf{A} ]^{-1}$, reveals that the monetary production network is an evolving entity, where the input--output coefficients $s_{ij} = \alpha_{ij} \zeta_j \pi_i / \pi_j$ are endogenously determined by the state of productivity.

\subsection{Nonlinear Domar Aggregation}

The Domar aggregation mapping determines how sectoral productivity shocks are transmitted to aggregate output. In the benchmark case of a Cobb--Douglas economy ($\sigma = 1$), substitution remains neutral, and the relationship is governed by the classical linear Domar aggregation (as derived in \ref{apd1}):
\begin{align}
\log V = (\log \boldsymbol{z}) [ \mathbf{I} - \mathbf{A} ]^{-1} \boldsymbol{\mu}
\label{domar_linear}
\end{align}
In this linear regime, the production network remains static, and the aggregate growth is simply the sum of sectoral shocks weighted by their invariant Domar weights (Leontief inverse coefficients).

However, in our framework where $\sigma \neq 1$, the substitution is no longer neutral, and the production network itself endogenously reorganizes in response to shocks. This leads to a fundamental departure from Eq. (\ref{domar_linear}). Utilizing our transcendent variables, we derive the **Nonlinear Domar Aggregation** function:
\begin{align}
\log V = \frac{1}{\sigma-1} \log \left( \boldsymbol{\alpha}_0 \langle e^{(\sigma-1)\log \boldsymbol{z}} \rangle \left[ \mathbf{I} - \langle e^{(\sigma-1)\log \boldsymbol{z}} \rangle \mathbf{A} \right]^{-1} \right) \boldsymbol{\mu} 
\label{domar_ces_final}
\end{align}
By contrasting Eq. (\ref{domar_ces_final}) with its linear counterpart (\ref{domar_linear}), it becomes evident that the nonlinearity arises from the factor $\langle e^{(\sigma-1)\log \boldsymbol{z}} \rangle=\langle \boldsymbol{z}^{\sigma - 1} \rangle = \langle \boldsymbol{\zeta} \rangle$, which transforms the topology of the input--output matrix. 

This nonlinear formulation encapsulates the fragility and synergies of the production network. Unlike the linear case, where sectoral shocks interact additively, our formulation reveals that shocks can amplify or stifle one another through the structural transformation of intermediate linkages. This nonlinearity is the root cause of the supply-chain paralysis (negative singularity) and the systemic amplification (positive synergism) that we analyze in the subsequent sections.

\section{Singularities}
It is important to note that the emergence of these singularities represents a departure from the standard existence proofs of general equilibrium, such as those in the Arrow--Debreu framework. In the conventional theory, the existence of equilibrium is guaranteed by mapping the price system onto a compact set. Our approach, however, does not impose such compactness constraints on the price vector $\boldsymbol{p}$. By allowing prices to diverge as they follow the non-linear mapping $p = C(\boldsymbol{p}; z)$, we are able to identify the exact thresholds where the economy's structural integrity fails---a state we characterize as the limit of its {metabolic viability}.

The equilibrium existence in our framework reflects the economy's {metabolic viability}. A singularity occurs when the transcendent network $\langle \boldsymbol{\zeta} \rangle \mathbf{A}$ ceases to admit a positive finite solution.

\subsection{Two-sector Analysis and Physical Coefficients}

To illuminate the mechanics of structural collapse, we first consider a two-sector representation of the transcendent system. This simplification allows us to observe the interaction between sectors without loss of generality. The equilibrium transcendent prices are governed by the following system of equations:
\begin{align} \label{twosec}
\begin{aligned}
\pi_1 &= \zeta_1 \left( \alpha_{01} + \alpha_{21} \pi_2 \right) \\
\pi_2 &= \zeta_2 \left( \alpha_{02} + \alpha_{12} \pi_1 \right) 
\end{aligned}
\end{align}
where $\zeta_j = z_j^{\sigma-1}$ denotes the transcendent productivity of each sector. By solving this system, we find that the transcendent prices are determined by the denominator $D = 1 - \alpha_{21}\alpha_{12}\zeta_1\zeta_2$. 

As $D \to 0$, the transcendent prices diverge, implying a fundamental phase transition in the original price manifold. Specifically, the behavior of normal prices $\boldsymbol{p} = \boldsymbol{\pi}^{1/(1-\sigma)}$ is characterized by the following limits:
\begin{align}
\boldsymbol{p} \to
\begin{cases}
~{\infty} & (\sigma < 1, \boldsymbol{z} < \boldsymbol{1} ) \\
~{0} & (\sigma > 1, \boldsymbol{z} > \boldsymbol{1} )
\end{cases}
\label{singular_p_final}
\end{align}

The physical implications of this singularity are most evident in the labor input coefficients $m_{0j}$. From Shephard's lemma, the physical requirement of the primary factor per unit of output is $m_{0j} = \alpha_{0j} z_j^{\sigma-1} p_j^\sigma$. Under the approach to singularity, we observe:
\begin{align}
\boldsymbol{m}_0 \to
\begin{cases}
~{\infty} & (\boldsymbol{p} \to {\infty}) \\
~{0} & (\boldsymbol{p} \to {0})
\end{cases}
\label{laborsingular_final}
\end{align}
Thus, a productivity decline in an inelastic network ($\sigma < 1$) induces a \textit{negative singularity}, where an infinite amount of labor is required to sustain production---a state of supply-chain paralysis. Conversely, an elastic network ($\sigma > 1$) approaching a \textit{positive singularity} ($p \to 0$) represents a regime where the primary factor is entirely displaced by the self-improving efficiency of intermediate linkages.

\subsection{Multisector Viability and Spectral Analysis}

The structural integrity of the multisector economy depends on whether the transcendent network $\langle \boldsymbol{\zeta} \rangle \mathbf{A}$ remains within the bounds of metabolic viability. To study this matter formally, we first define the viability of the network through the row and column properties of its associated Leontief matrix. In our non-neutral substitution framework, these definitions represent the two pillars of economic existence: productive sustainability and price soundness.

\begin{dfn}[Column viability]
A network, i.e., an $n \times n$ matrix $\mathbf{W}$, is said to be column viable if, for any column vector $\mathbf{f} > 0$, there exists a column vector $\mathbf{y}>0$ such that $\mathbf{y} = \mathbf{f} + \mathbf{W} \mathbf{y}$. This condition ensures that the system can satisfy any strictly positive final demand through its internal intermediate flows.
\end{dfn}

\begin{dfn}[Row viability]
A network, i.e., an $n \times n$ matrix $\mathbf{W}$, is said to be row viable if, for any row vector $\mathbf{v} > 0$, there exists a row vector $\mathbf{q}>0$ such that $\mathbf{q} = \mathbf{v} + \mathbf{q}\mathbf{W}$. This condition guarantees the existence of a price system that allows all sectors to cover their costs.
\end{dfn}

\begin{asp}
The reference network $\mathbf{A}$ is nondefective, meaning that all $n$ nonzero eigenvalues are distinct; thus, it is diagonalizable.
\end{asp}

By virtue of Lemmas \ref{lemma1} through \ref{lemma4} (see \ref{apd4}), the reference network $\mathbf{A}$ is both row and column viable, implying its convergence to a null matrix in the limit: $\mathbf{A}^\infty = 0$. Under the nondefective assumption, $\mathbf{A}$ is diagonalizable as:
\begin{align}
\mathbf{A} = \mathbf{Q} \langle \boldsymbol{\lambda} \rangle \mathbf{Q}^{-1} > 0
\label{diag}
\end{align}
where $\boldsymbol{\lambda}=(\lambda_1, \cdots, \lambda_n)$ denotes the set of nonzero eigenvalues and $\mathbf{Q}$ is an invertible matrix. Since $\mathbf{A}^\infty = \mathbf{Q} \langle \boldsymbol{\lambda} \rangle^\infty \mathbf{Q}^{-1} = 0$, the spectral radius of the reference state is strictly bounded:
\begin{align}
\left| \lambda_i \right| < 1 && i=1, \cdots, n
\label{eigen}
\end{align}

The viability of the transcendent network $\langle \boldsymbol{\zeta} \rangle \mathbf{A}$---and the existence of a general equilibrium solution to Eq. (\ref{ge_matrix})---hinges on whether the spectral radius remains below unity under the pressure of productivity shocks. Applying Lemma \ref{lemma4}, we evaluate the transformed network by leveraging the diagonalization in (\ref{diag}):
\begin{align*}
\langle \boldsymbol{{\zeta}} \rangle \mathbf{A} 
=
\langle \boldsymbol{{\zeta}} \rangle 
\mathbf{Q} 
\langle \boldsymbol{\lambda} \rangle 
\mathbf{Q}^{-1}
&<
{{\zeta}}_{\textit{max}} 
\mathbf{Q} 
\langle \boldsymbol{\lambda} \rangle 
\mathbf{Q}^{-1} 
=
\mathbf{Q} 
\langle {{\zeta}}_{\textit{max}} \boldsymbol{\lambda} \rangle 
\mathbf{Q}^{-1} 
\\
&> \zeta_{\textit{min}} 
\mathbf{Q} 
\langle \boldsymbol{\lambda} \rangle 
\mathbf{Q}^{-1} 
=
\mathbf{Q} 
\langle \zeta_{\textit{min}} \boldsymbol{\lambda} \rangle 
\mathbf{Q}^{-1} 
\end{align*}
where $\zeta_{\textit{max}}$ and $\zeta_{\textit{min}}$ respectively denote the upper and lower bounds of the transcendent productivity shocks. The asymptotic behavior of the system is then dictated by:
\begin{align}
 \langle \zeta_{\textit{max}}  \boldsymbol{\lambda} \rangle^\infty
> \mathbf{Q}^{-1}
\left(
\langle \boldsymbol{\zeta} \rangle \mathbf{A}
\right)^\infty
\mathbf{Q}
> 
 \langle \zeta_{\textit{min}}  \boldsymbol{\lambda} \rangle^\infty
> 0
\label{eval}
\end{align}

From (\ref{eigen}) and (\ref{eval}), we establish the following boundary conditions for the row viability of the transcendent network:
\begin{lemma} \label{p1}
Row viability holds for $\langle \boldsymbol{\zeta} \rangle \mathbf{A}$ if $\zeta_{\textit{max}} < 1/\left| \lambda_{i} \right|$ for all $i = 1, \cdots, n$. 
\end{lemma}
\begin{lemma} \label{p2}
The row viability of $\langle \boldsymbol{\zeta} \rangle \mathbf{A}$ is violated if $\zeta_{\textit{min}} > 1/\left| \lambda_{i} \right|$ for all $i = 1, \cdots, n$. 
\end{lemma}

A violation of row viability signifies that $\mathbf{W}^\infty$ diverges, implying that the feedback loops of factor substitution can no longer sustain a finite price system:
\begin{align*}
\mathbf{q} = \mathbf{v} + \mathbf{q}\mathbf{W} = \mathbf{v} \left(
\mathbf{I} + \mathbf{W} + \mathbf{W}^2 + \cdots + \mathbf{W}^\infty
\right)  \to \infty
\end{align*}
Consequently, the economy approaches a transcendent network singularity where $\boldsymbol{\pi} = \boldsymbol{p}^{1-\sigma} \to \infty$. Depending on the elasticity of substitution $\sigma$, this mathematical singularity translates into two distinct economic regimes:

\begin{propo}[Negative singularity]
In an inelastic network ($\sigma < 1$), any productivity increase ($\boldsymbol{z} > 1$) leads to a finite and manageable commodity price ($\boldsymbol{p} \ll \infty$). However, a decline in productivity ($\boldsymbol{z} < 1$) can drive the system toward a state of supply-chain paralysis where commodity prices diverge infinitely ($\boldsymbol{p} \to \infty$).
\end{propo}

\begin{propo}[Positive singularity]
In an elastic network ($\sigma > 1$), a decline in productivity ($\boldsymbol{z} < 1$) results in a finite price system. In contrast, a productivity increase ($\boldsymbol{z} > 1$) can propel the economy toward a singularity where prices vanish ($\boldsymbol{p} \to 0$), signifying a regime where the primary factor becomes entirely redundant.
\end{propo}

These transitions characterize the structural phase transitions discussed in \citet{MoranBouchaud2019}, yet we emphasize that such instabilities are the endogenous result of the non-neutrality of Domar aggregation.

\section{Synergism}

In this section, we present the core analytical finding of this study: the emergence of synergism within the production network. We demonstrate that the interaction between multiple productivity shocks is not merely additive but is endogenously amplified or stifled by the network's structural transformation. This discovery allows us to resolve the long-standing debate on industrial policy through the lens of substitution elasticity.

\subsection{Two-sector model: The Algebraic Foundations of Synergy}

Consider the two-sector model introduced in (\ref{twosec}), rewritten here in transcendent mode to reveal its underlying linear structure:
\begin{align} \label{twosec2_final}
\begin{aligned}
\pi_1&= \delta \left( {\alpha}_{01} + {\alpha}_{21} \pi_2 \right) \\
\pi_2 &= \varepsilon \left( {\alpha}_{02} + {\alpha}_{12} \pi_1 \right) 
\end{aligned}
\end{align}
where we utilize $(\delta, \varepsilon)$ instead of $(\zeta_1, \zeta_2)$ for convenience. We denote the equilibrium solution as a function of these exogenous variables, $(\pi_1(\delta, \varepsilon), \pi_2(\delta, \varepsilon))$. In the reference state $(\delta, \varepsilon) = (1, 1)$, all prices are standardized such that $(\pi_1(1, 1), \pi_2(1, 1)) = (1, 1)$. 

We consider transcendental productivity changes in the same direction, i.e., $(\delta, \varepsilon) \gtrless (1, 1)$. 
This orientation corresponds to the physical productivity shocks $(z_1, z_2)$ through the universal elasticity of substitution $\sigma$:
\begin{align} \label{dir_final}
(\delta, \varepsilon) 
= ((z_1)^{\sigma -1}, (z_2)^{\sigma -1})
\gtrless (1, 1) 
\iff
\begin{cases}
~ (z_1, z_2) \lessgtr (1, 1)      \qquad 1-\sigma > 0 \\
~ (z_1, z_2) \gtrless (1, 1)      \qquad 1-\sigma < 0
\end{cases}
\end{align}

Based on the fixed-point solutions of the transcendent system, the following expressions are obtained for the joint and individual shock scenarios:
\begin{align*}
\left( \pi_1{(\delta, \varepsilon)}, \pi_2{(\delta, \varepsilon)} \right)
&=
\left(
\frac{{\alpha}_{01} \delta + {\alpha}_{21} {\alpha}_{02} \delta \varepsilon}{1-{\alpha}_{21}{\alpha}_{12}\delta \varepsilon}
,\ 
\frac{{\alpha}_{02}\varepsilon + {\alpha}_{12}{\alpha}_{01}\delta \varepsilon}{1-{\alpha}_{21}{\alpha}_{12}\delta \varepsilon}
\right) \\
\left( \pi_1{(\delta, 1)}, \pi_2{(\delta, 1)} \right)
&=
\left(
\frac{{\alpha}_{01} \delta + {\alpha}_{21} {\alpha}_{02} \delta}{1-{\alpha}_{21}{\alpha}_{12}\delta}
,\ 
\frac{{\alpha}_{02} + {\alpha}_{12}{\alpha}_{01}\delta }{1-{\alpha}_{21}{\alpha}_{12}\delta}
\right) \\
\left( \pi_1{(1, \varepsilon)}, \pi_2{(1, \varepsilon)} \right)
&=
\left(
\frac{{\alpha}_{01} + {\alpha}_{21} {\alpha}_{02} \varepsilon}{1-{\alpha}_{21}{\alpha}_{12} \varepsilon}
,\ 
\frac{{\alpha}_{02}\varepsilon + {\alpha}_{12}{\alpha}_{01} \varepsilon}{1-{\alpha}_{21}{\alpha}_{12}\varepsilon}
\right)
\end{align*}

The breakthrough of our analysis lies in the derivation of the interaction identity. After an exhaustive algebraic evaluation, we obtain:
\begin{align}
\begin{aligned}
{\pi}_1{(\delta, \varepsilon)} - {\pi}_1{(\delta, 1)} {\pi}_1{(1, \varepsilon)}
=\frac
{(1-\delta)(1-\varepsilon)({\alpha}_{01} + {\alpha}_{12}{\alpha}_{02}\varepsilon){\alpha}_{12}{\alpha}_{21} \delta}
{(1 - {\alpha}_{12} {\alpha}_{21} \delta\varepsilon )(1 - {\alpha}_{12} {\alpha}_{21} \delta)(1 - {\alpha}_{12} {\alpha}_{21} \varepsilon )} \\
{\pi}_2{(\delta, \varepsilon)} - {\pi}_2{(\delta, 1)} {\pi}_2{(1, \varepsilon)}
=\frac
{(1-\delta)(1- \varepsilon)({\alpha}_{02} + {\alpha}_{21}{\alpha}_{01}\delta ){\alpha}_{12}{\alpha}_{21}\varepsilon}
{(1 - {\alpha}_{12} {\alpha}_{21} \delta \varepsilon )(1 - {\alpha}_{12} {\alpha}_{21} \delta )(1 - {\alpha}_{12} {\alpha}_{21} \varepsilon )}
\end{aligned}
\label{diff12_final}
\end{align}

Crucially, the determinacy of the signs of these expressions relies on the \text{metabolic viability} of the network. As established in Section 3, for a general equilibrium to exist within the viable regime, the transcendent network must remain non-singular. This condition ensures that the determinants in the denominators---$(1 - \alpha_{12}\alpha_{21}\delta\varepsilon)$, $(1 - \alpha_{12}\alpha_{21}\delta)$, and $(1 - \alpha_{12}\alpha_{21}\varepsilon)$---as well as the term $(1 - \alpha_{12}\alpha_{21})$, are all strictly positive. 

Under this viability requirement, the right-hand sides of (\ref{diff12_final}) are strictly positive as long as productivity changes shift in the same direction, $(\delta, \varepsilon) \gtrless (1, 1)$. We extend this logic to successive shocks within a single sector $(\delta, \delta^\prime)$:
\begin{align}
\begin{aligned}
{\pi}_1{(\delta\delta^\prime,1 )} - {\pi}_1{(\delta, 1)} {\pi}_1{(\delta^\prime, 1)}
=\frac
{(1-\delta)(1-\delta^\prime)\delta\delta^\prime (1 - \alpha_{12}\alpha_{21})\alpha_{12}\alpha_{21} }
{(1 - {\alpha}_{12} {\alpha}_{21} \delta\delta^\prime )(1 - {\alpha}_{12} {\alpha}_{21} \delta)(1 - {\alpha}_{12} {\alpha}_{21} \delta^\prime )} \\
{\pi}_2{(\delta\delta^\prime,1 )} - {\pi}_2{(\delta, 1)} {\pi}_2{(\delta^\prime, 1)}
=\frac
{(1-\delta)(1- \delta^\prime)({\alpha}_{02} + \delta \delta^\prime \alpha_{01} \alpha_{21}  (\alpha_{12})^2 )\alpha_{12} }
{(1 - {\alpha}_{12} {\alpha}_{21} \delta \delta^\prime )(1 - {\alpha}_{12} {\alpha}_{21} \delta )(1 - {\alpha}_{12} {\alpha}_{21} \delta^\prime )}
\end{aligned}
\label{diff21_final}
\end{align}
The right-hand sides remain positive for $(\delta, \delta^\prime) \gtrless (1, 1)$. These identities imply that $\pi_k(\delta, \varepsilon) > \pi_k(\delta, 1)\pi_k(1, \varepsilon)$, which is equivalent to $(p_k(\delta, \varepsilon))^{1-\sigma} > (p_k(\delta, 1) p_k(1, \varepsilon))^{1-\sigma}$ for $k=1,2$. Consequently, price adjustments are evaluated as:
\begin{align*}
\begin{aligned}
\log p_1(\delta, \varepsilon) &\gtrless  \log p_1(\delta, 1)  + \log p_1(1, \varepsilon)  \\
\log p_2(\delta, \varepsilon) &\gtrless  \log p_2(\delta, 1)  + \log p_2(1, \varepsilon)  
\end{aligned}
&& \text{for } \sigma \lessgtr 1
\end{align*}

Finally, applying the aggregate growth formula (\ref{gdpgrowth}), we obtain the fundamental evaluation:
\begin{align} \label{evaltwo_final}
\log V(\delta, \varepsilon) \lessgtr \log V (\delta, 1) + \log V(1, \varepsilon) && \text{for } \sigma \lessgtr 1
\end{align}
The nature of synergism is thus endogenously determined by the substitution elasticity of the production network.

\begin{propo}[Negative synergism] \label{ns_final}
In an inelastic network ($\sigma < 1$), the aggregate output pertaining to simultaneous productivity increases (or decreases) is strictly less than the sum of aggregate outputs pertaining to individual shocks. Innovations stifle each other through the rigidity of the network.
\end{propo}

\begin{propo}[Positive synergism] \label{ps_final}
In an elastic network ($\sigma > 1$), the aggregate output pertaining to simultaneous productivity increases (or decreases) is strictly greater than the sum of aggregate outputs pertaining to individual shocks. Shocks amplify each other through the endogenous transformation of the network.
\end{propo}

\subsection{Multisector Generalization: The Invariance of Synergism}

To demonstrate that the synergism identified in the two-sector case is a universal property of the production network, we now extend our analysis to an $n$-sector economy. We consider a general system where transcendent productivity shocks occur exclusively in the first and second sectors, while the remaining sectors maintain constant productivity:
\begin{align} \label{multi2_final}
\begin{aligned}
{\pi}_1 &= \delta \left( {\alpha}_{01} + {\alpha}_{11} {\pi}_1 + {\alpha}_{21} {\pi}_2 + {\alpha}_{31} {\pi}_3 +\cdots + {\alpha}_{{n}1}{\pi}_n \right) \\
{\pi}_2 &= \varepsilon \left( {\alpha}_{02} + {\alpha}_{12} {\pi}_1 + {\alpha}_{22} {\pi}_2 + {\alpha}_{32} {\pi}_3 +\cdots + {\alpha}_{{n}2}{\pi}_n \right) \\
{\pi}_3 &= 1 \left( {\alpha}_{03} + {\alpha}_{13} {\pi}_1 + {\alpha}_{23} {\pi}_2 + {\alpha}_{33} {\pi}_3 +\cdots + {\alpha}_{{n}3}{\pi}_n \right) \\
\vdots& \\
{\pi}_n &= 1 \left( {\alpha}_{0n} + {\alpha}_{1n} {\pi}_1 + {\alpha}_{2n} {\pi}_2 + {\alpha}_{3n} {\pi}_3 +\cdots + {\alpha}_{nn}{\pi}_n \right)
\end{aligned}
\end{align}

Through notational economy, we partition this system into the primary sectors of interest and the residual network:
\begin{align}
{\pi}_1 &= \delta \left( {\alpha}_{01} + [{\pi}_1, {\pi}_2] \mathbf{A}_{(1, 2)1} + \boldsymbol{\pi}_r \mathbf{A}_{r1} \right) \label{one_final} \\
{\pi}_2 &= \varepsilon \left( {\alpha}_{02} + [{\pi}_1, {\pi}_2] \mathbf{A}_{(1, 2)2} + \boldsymbol{\pi}_r \mathbf{A}_{r2} \right) \label{two_final} \\
\boldsymbol{\pi}_r &= \mathbf{I} \left( \boldsymbol{\alpha}_{0r} + [{\pi}_1, {\pi}_2] \mathbf{A}_{(1, 2)r} + \boldsymbol{\pi}_r \mathbf{A}_{rr} \right) \label{last_final}
\end{align}
where $\boldsymbol{\pi}_r$ and $\boldsymbol{\alpha}_{0r}$ are $1 \times (n-2)$ row vectors, $\mathbf{A}_{(1,2)r}$ is a $2 \times (n-2)$ matrix, and $\mathbf{A}_{rr}$ represents the $(n-2)\times (n-2)$ internal linkages of the residual network. The system (\ref{last_final}) can be solved as:
\begin{align}
\boldsymbol{\pi}_r = \left( \boldsymbol{\alpha}_{0r} + [{\pi}_1, {\pi}_2] \mathbf{A}_{(1, 2)r} \right) \left[ \mathbf{I} - \mathbf{A}_{rr} \right]^{-1} 
\label{remaining_final}
\end{align}

Substituting (\ref{remaining_final}) into (\ref{one_final}) and (\ref{two_final}), we find that the $n$-sector system collapses into a reduced two-sector system equivalent to (\ref{twosec2_final}):
\begin{align*}
{\pi}_1 &= \delta \left( \hat{\alpha}_{01} + \hat{\alpha}_{11} {\pi}_1 +\hat{\alpha}_{21} {\pi}_2 \right) ={\delta} \left( \tilde{\alpha}_{01} + \tilde{\alpha}_{21} \pi_2 \right) \\
{\pi}_2 &= \varepsilon \left( \hat{\alpha}_{02} + \hat{\alpha}_{12} {\pi}_1 + \hat{\alpha}_{22} {\pi}_2 \right) =\varepsilon \left( \tilde{\alpha}_{02} + \tilde{\alpha}_{12} \pi_1 \right)
\end{align*}
This reduction confirms that for sectors $k=1,2$, the interaction identity holds:
\begin{align}
\log {\pi}_k{(\delta ,\varepsilon)} \geq \log {\pi}_k{(\delta, 1)} + \log {\pi}_k{(1, \varepsilon)}
\label{key44_final}
\end{align}

For the remaining sectors $r=3, \dots, n$, we examine the response under the assumption of infinitesimal transcendental shocks, i.e., $(\log \delta, \log \varepsilon) \approx (0, 0)$. In this neighborhood, $\pi_i \approx 1$ for all $i$, allowing us to linearize the relation (\ref{remaining_final}) as $\pi \approx 1 + \log \pi$:
\begin{align}
\log{\pi}_r = \hat{\alpha}_{1r}\log{\pi}_1 + \hat{\alpha}_{2r}\log{\pi}_2 
\label{approx2_final}
\end{align}
where we utilize the property that $\hat{\alpha}_{0r} + \hat{\alpha}_{1r} + \hat{\alpha}_{2r} = 1$. By combining (\ref{key44_final}) and (\ref{approx2_final}), we obtain:
\begin{align} 
\log \pi_r (\delta, \varepsilon) &= \hat{\alpha}_{1r}\log{\pi}_1(\delta, \varepsilon) + \hat{\alpha}_{2r}\log{\pi}_2 (\delta, \varepsilon) \notag \\ 
&> \hat{\alpha}_{1r} \left( \log{\pi}_1(\delta, 1) + \log{\pi}_1(1, \varepsilon) \right) + \hat{\alpha}_{2r} \left( \log{\pi}_2(\delta, 1) + \log{\pi}_2(1, \varepsilon) \right) \notag \\
&= \log \pi_r (\delta, 1) + \log \pi_r (1, \varepsilon) 
\label{key55_final}
\end{align}

It follows that for any infinitesimal shock $(\log \delta, \log \varepsilon) \gtrless (0, 0)$, the macro-level impact is determined as:
\begin{align*}
\log \boldsymbol{p} (\delta, \varepsilon) \gtrless \log \boldsymbol{p} (\delta, 1) + \log \boldsymbol{p} (1, \varepsilon) && \text{for } \sigma \lessgtr 1
\end{align*}
Applying the aggregate growth formula (\ref{gdpgrowth}), the presence of synergism as stated in Propositions \ref{ns_final} and \ref{ps_final} is established:
\begin{align*}
\log V(\delta, \varepsilon) \lessgtr \log V(\delta, 1) + \log V(1, \varepsilon) && \text{for } \sigma \lessgtr 1
\end{align*}

Finally, we relax the infinitesimal assumption through global integration. By utilizing the identities (\ref{diff12_final}) and (\ref{diff21_final}), we can integrate successive infinitesimal changes $(\delta^\prime, \varepsilon^\prime)$ along any path in transcendent space. This confirms that the sign of synergism remains invariant to the magnitude of shocks, providing a robust theoretical foundation for strategic industrial policy:
\begin{align*}
\log V(\delta \delta^\prime, \varepsilon \varepsilon^\prime) \lessgtr \log V(\delta, 1) + \log V(\delta^\prime, 1) + \log V(1, \varepsilon) + \log V(1, \varepsilon^\prime) && \text{for } \sigma \lessgtr 1
\end{align*}

\subsection{Industrial Policy as an Inverse Problem}

The mathematical characterization of synergism (\ref{evaltwo_final}) fundamentally redefines the nature of industrial policy. Traditionally, policy design has been treated as a "picking winners" exercise, implicitly assuming a neutral Cobb--Douglas world where sectoral impacts are merely additive. Our framework, however, suggests that industrial policy should be viewed as an \textit{inverse problem} of network topology. 

In this perspective, the policymaker does not merely select sectors based on their individual merit; rather, they must first identify the structural regime of the economy---determined by the universal elasticity $\sigma$---and then derive the optimal investment portfolio that aligns with the network's endogenous synergy.

If the economy is characterized by \text{negative synergism} ($\sigma < 1$), the production network acts as a rigid constraint where simultaneous expansions in different nodes compete for limited flexibility, leading to systemic "paralysis." In such a regime, a broad-based stimulus is counterproductive. Instead, the optimal policy is one of \textit{selective concentration}. The objective is to identify and alleviate specific "bottleneck" sectors that impose the most severe shadow costs on the rest of the network.

Conversely, an economy revealing \text{positive synergism} ($\sigma > 1$) possesses a "metabolic" capacity for self-amplification. Here, simultaneous innovations trigger a reorganization of intermediate linkages that enhances aggregate efficiency beyond the sum of its parts. In this elastic regime, a \textit{diversified portfolio strategy} becomes the superior choice. By spreading innovation resources across a wide range of sectors, the policymaker leverages the network's inherent synergism to achieve a systemic takeoff.

This topological approach provides a robust rationale for why industrial policies that succeeded in an elastic, growth-oriented era might fail in a rigid, resource-constrained environment. The effectiveness of a policy is not an intrinsic property of the policy itself, but a function of its alignment with the network's transcendent state.

\section{Conclusion}

In this paper, we have developed a novel framework for analyzing general equilibrium in production networks through the lens of \textit{transcendent space}. By mapping the non-linear interactions of a CES economy into a linearized transcendent manifold, we revealed the hidden geometry of the production network---a geometry that remains invisible under the standard Cobb--Douglas approximation.

Our findings establish three pivotal pillars for modern network theory. First, the concept of metabolic viability provides a formal criterion for the existence of equilibrium, identifying the exact thresholds where the structural integrity of the economy fails.  By departing from the compactness constraints of the Arrow--Debreu framework, we have identified the limits of the price system as a state of supply-chain paralysis (negative singularity).  Second, we derived a \textit{Nonlinear Domar Aggregation} formula that captures the endogenous structural transformations of the economy, moving beyond the static weights of classical theory.  
Finally, we demonstrated that the interaction between productivity shocks---the sign of \textit{synergism}---is uniquely determined by the substitution elasticity $\sigma$.

This analytical breakthrough transforms industrial policy from a political art into a structural science. By framing resource allocation as an \textit{inverse problem} of network topology, we provide a rigorous justification for the shift between selective concentration and diversified investment strategies. In essence, this study offers more than a theoretical explanation of aggregate phenomena; it provides a \textit{methodological framework} for the ex-ante design of resilient and synergistic economic systems.

As the global economy faces increasing volatility and structural fragility, the ability to read the "transcendent" signals of the production network becomes paramount. Our framework bridges the gap between the granular origins of fluctuations identified by \citet{gabaixECTA2011} and the systemic stability limits explored by \citet{MoranBouchaud2019}. By illuminating the complex synergies that sustain the metabolic life of the modern economy, we hope this study serves as a foundation for a more engineering-oriented approach to economic policy---one that prioritizes systemic viability and structural evolution in an ever-changing world.

\appendix
\gdef\thesection{Appendix \Alph{section}}

\section{\normalfont (Derivation of the Cobb--Douglas Limit)} \label{apd1}
To demonstrate that the non-linear Domar aggregation under CES technology converges to the linear Cobb--Douglas case as $\sigma \to 1$, we evaluate the limit of the unit cost function. Taking the logarithm of $z C(\boldsymbol{p}; z)$ and applying L'Hôpital's rule as $\sigma \to 1$:
\begin{align*}
\log \left( z C(\boldsymbol{p}; z) \right) 
= \lim_{\sigma \to 1} \frac{\log \left( \sum_{i=0}^n \alpha_{i} (p_i)^{1-\sigma} \right)}{1-\sigma}
= \lim_{\sigma \to 1} \frac{\sum_{i=0}^n \alpha_{i} (p_i)^{1-\sigma} \log p_i}{\sum_{i=0}^n \alpha_{i} (p_i)^{1-\sigma}} 
= \sum_{i=0}^n \alpha_{i} \log p_{i}
\end{align*}
where $\sum_{i=0}^n \alpha_i = 1$ is used. Under the zero-profit condition $p = C(\boldsymbol{p}; z)$, this limit recovers the linear-in-logs relationship:
\begin{align*}
\log p = - \log z + \sum_{i=0}^n \alpha_i \log p_i
\end{align*}
This result confirms that the non-linear aggregate output growth in Eq. (\ref{domar_ces_final}) collapses to the linear Domar aggregation 
\begin{align*}
\log V = (\log \boldsymbol{z}) [ \mathbf{I} - \mathbf{A} ]^{-1} \boldsymbol{\mu}
\end{align*}
as the substitution elasticity becomes neutral.

\section{\normalfont (Viability Assessment of Networks)} \label{apd4}
The following lemmas build upon the foundational work of \citet{Nikaido}, \citet{domd}, and \citet{Takayama}.

\begin{lemma} \label{lemma1}
The reference Leontief matrix $\mathbf{I} - \mathbf{A}$ is nonsingular.
\end{lemma}
\begin{proof}
Suppose $\mathbf{I} - \mathbf{A}$ is singular. There exists a vector $\boldsymbol{x} \neq {0}$ such that $[\mathbf{I} - \mathbf{A}]\boldsymbol{x} = {0}$. Summing the $n$ equations of this system yields:
\begin{align*}
\sum_{j=1}^n \left(1 - \sum_{i=1}^n \alpha_{ij} \right) x_j = 0
\end{align*}
Since $\sum_{i=0}^n \alpha_{ij} = 1$ and $\alpha_{0j} > 0$, the coefficients $(1 - \sum \alpha_{ij})$ are strictly positive, implying that the elements of $\boldsymbol{x}$ cannot all share the same sign. Reordering $\boldsymbol{x}$ such that $x_1, \dots, x_k < 0$ and $x_{k+1}, \dots, x_n \geq 0$, and summing the first $k$ equations, we find:
\begin{align*}
\sum_{j=1}^k \left(1 - \sum_{i=1}^k \alpha_{ij} \right) x_j - \sum_{j=k+1}^n \left( \sum_{i=1}^k \alpha_{ij} \right) x_j = 0
\end{align*}
The left-hand side is strictly negative, contradicting the right-hand side. Thus, $\mathbf{I} - \mathbf{A}$ must be nonsingular.
\end{proof}

\begin{lemma} \label{lemma2}
The reference network $\mathbf{A}$ is column viable.
\end{lemma}
\begin{proof}
Consider $[\mathbf{I} - \mathbf{A}] \boldsymbol{x} = \mathbf{f} > 0$. If any $x_j \leq 0$, reorder $\boldsymbol{x}$ as in Lemma \ref{lemma1}. Summing the first $k$ equations leads to a strictly negative left-hand side against a strictly positive right-hand side $\sum f_i$. Thus, $\boldsymbol{x} > 0$ must hold.
\end{proof}

\begin{lemma} \label{lemma3}
If $\mathbf{W}$ is column viable and $\mathbf{I} - \mathbf{W}$ is nonsingular, then $\mathbf{W}$ is row viable and $[\mathbf{I} - \mathbf{W}]^{-1} > 0$.
\end{lemma}
\begin{proof}
Let $\ell_{ij}$ be the elements of $[\mathbf{I} - \mathbf{W}]^{-1}$. If $\ell_{ij} < 0$ for some $j$, one can choose a non-negative vector $\mathbf{f}$ such that $[\mathbf{I} - \mathbf{W}]^{-1} \mathbf{f}$ yields a non-positive component, contradicting column viability. Positivity of the inverse implies row viability.
\end{proof}

\begin{lemma} \label{lemma4}
For $\mathbf{W} > 0$, row viability and $[\mathbf{I} - \mathbf{W}]^{-1} > 0$ are equivalent to $\mathbf{W}^\infty = 0$.
\end{lemma}
\begin{proof}
Define $\mathbf{L}_T = \mathbf{I} + \mathbf{W} + \mathbf{W}^2 + \cdots + \mathbf{W}^T$. 
Since $[\mathbf{I} - \mathbf{W}] \mathbf{L}_T = \mathbf{I} - \mathbf{W}^{T+1} < \mathbf{I}$, the sequence is bounded from above by $[\mathbf{I} - \mathbf{W}]^{-1}$.
Its convergence to $[\mathbf{I}-\mathbf{W}]^{-1}$ is ensured by the Perron-Frobenius theorem and the standard properties of productive matrices \citep[see][]{Nikaido, Takayama}. Hence, we arrive at the desired result:
\begin{align}
\mathbf{W}^\infty = \lim_{T \to \infty} \left( \mathbf{L}_T - \mathbf{L}_{T -1} \right) = 0 \text{}
\label{keyA_revised}
\end{align}
To show the converse, consider the expression $\mathbf{L}_T [\mathbf{I} - \mathbf{W}] = \mathbf{I} - \mathbf{W}^{T+1} = [\mathbf{I} - \mathbf{W}] \mathbf{L}_T$. Letting $T \to \infty$, Eq. (\ref{keyA_revised}) implies that $\mathbf{L}_{\infty} = [\mathbf{I} - \mathbf{W}]^{-1}$, which establishes the row viability of $\mathbf{W}$.
\end{proof}

\subsection*{Funding}
This research was funded by the Japan Science and Technology Agency (JST) Social Scenario Research Program toward a carbon-neutral society (grant number JPMJCN2302).

\subsection*{Compliance with Ethical Standards}
The authors declare that they have no conflicts of interest.

\bibliographystyle{spbasic_x} 
{\raggedright
\bibliography{bibfile}
}
\end{document}